\documentclass[12pt,twoside]{article} 
\usepackage{geometry}
\usepackage[all]{xy}
\usepackage{amsmath,amsthm}
\usepackage{amssymb}
\usepackage{listings}
\usepackage{xcolor}
\usepackage{color}
\usepackage{graphicx}
\usepackage{verbatim}
\usepackage{epstopdf}
\usepackage{subfigure}
\usepackage{algorithm}
\usepackage{algorithmicx}
\usepackage{algpseudocode}
\usepackage[colorlinks,
            linkcolor=blue,
            anchorcolor=blue,
            citecolor=blue]{hyperref}

\date{} 

\begin{document}

\centerline{\Large{\bf Forward and Reverse Parking in a Parking Lot}} 

\centerline{} 

\centerline{\bf {Kexin XIE}} 

\centerline{} 

\centerline{School of Mathematics and Statistics}

\centerline{Dalian University of Technology} 

\centerline{Dalian 116024, Liaoning, China} 

\centerline{} 

\centerline{\bf {Myron HLYNKA}} 

\centerline{} 

\centerline{Department of Mathematics and Statistics} 

\centerline{University of Windsor} 

\centerline{Windsor, Ontario, Canada N9B 3P4} 

\newtheorem{Theorem}{\quad Theorem}[section] 

\newtheorem{Definition}[Theorem]{\quad Definition} 

\newtheorem{Corollary}[Theorem]{\quad Corollary} 

\newtheorem{Lemma}[Theorem]{\quad Lemma} 

\newtheorem{Example}[Theorem]{\quad Example} 

\centerline{}

\begin{abstract}
The choice of forward and reverse parking in a parking lot is studied as a stochastic process. An $M/M/c/c$ queueing system is used as an initial framework. We use Monte Carlo simulation to get the relationship between vehicle orientation and vehicle entry and exit rates, as well as the most likely parking states at each specific rate. We view the change in parking status over time. 
\end{abstract}

{\bf Mathematics Subject Classification:} 60K25, 90B22 \\

{\bf Keywords:} 
queueing system, simulation, parking lot, reverse parking, forward parking

\section{Introduction}
There is considerable mathematical research on parking lots. See \cite{CalM}, \cite{Nour}, for example.
However, there does not appear to be a mathematical study on reverse parking versus forward parking, in parking lots.  

When people examine a parking lot, they will find that there are two possible parking directions for the vehicle: forward parking and reverse parking (or backward parking). 

\begin{figure}[htb]
   \centering
   \includegraphics[scale=0.3]{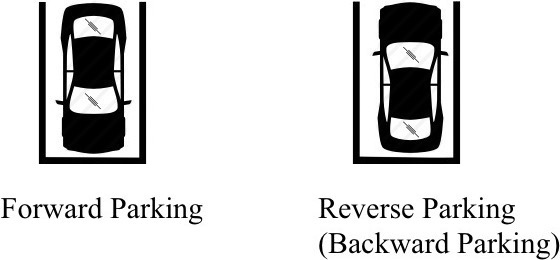}
\end{figure}

Reverse parking is considered a safer way to forward parking. From the USA FIRE PROTECTION official website  (\cite{mayK}) ``Reverse parking is about making the environment safer when the driver leaves the parking space. When reverse parking, a driver is going into a known space with no vehicle and pedestrian traffic. When leaving the parking space the driver is able to see the surroundings more clearly". The website thejournal.ie  (\cite{MelM}) suggests that people should reverse park because reversing into a parking space gives you greater control and makes it easier to maneuver out of the space. ``If you think reverse parking is difficult, it's much more difficult to try and wiggle your car out of a tight space when you have parked front first."

The City of Mount Pearl, Newfoundland adopted a Reverse Parking Policy (\cite{CMP})  It says `` all employees parking vehicles on City of Mount Pearl property and on any other parking lot must ensure to park the vehicle safely so you can exit the parking lot by pulling out head-on."

On the other hand, in many cases, people are not willing to reverse park. ``In a study conducted by The National Highway Traffic Safety Administration (NHTSA), on average $76\%$ or drivers park nose-in, the way many of us are used to parking." (\cite{mayK}) This situation makes sense to some extent. For example, in a double-row parallel parking lot, people have a high probability of not parking backward when the opposite parking space is already occupied; in a single-row parking lot, people are less willing to reverse park because it would take more time, and slow down or vehicles directly behind. .
 
In this article, we will build a mathematical model to simulate the parking directions in the parking lot and discuss the relation between parking direction and arrival time and service time. We also use Markov Chain methods to determine the effectiveness of our orientation model. This article appears to be the first mathematical model to study reverse parking in parking lots. 

Recall the density and cumulative distribution function (cdf) of an exponential distribution with parameter $\lambda$ are given by $f(t)=\lambda e^{-\lambda t}, \quad t>0$ and $F(t)=1-e^{-\lambda t}, \quad t\ge 0.$\\
Here $E(X)=\frac{1}{\lambda}.$ An important property of an exponential random variable $X$ is the \textit{$memoryless \ property$}. This property states that for all $x\ge0$ and $t\ge0$,$P(X>x+t|X>t)=P(X>x).$ This property is used in our simulation to simplify the study of the next event. Beyond the \textit{$memoryless \ property$}, we also use the well-known properties below (\cite{FellW}) , in our simulation programming.
 
\begin{Theorem}\label{T:min} Let $W$,$X$,$Y$ be three random variables and $X$,$Y$ are independent of each other. If $X\sim Exp(\lambda)$, $Y\sim Exp(\mu)$, then $$W=\min (X,Y)\sim Exp(\lambda+\mu).$$
\end{Theorem}

\begin{Theorem}\label{T:small} Let $X$,$Y$ be independent and exponentially distributed random variables with respective parameters $\lambda$ and $\mu$, then $$P(X<Y)=\frac{\lambda}{\lambda+\mu}.$$\end{Theorem}

\section{Mathematical Model}
In this section, we will describe the parking lot model. We assume that the parking lot is an M/M/c/c queueing system. i.e. We assume that vehicles arrive according to a Poisson process at rate $\lambda$, with independent exponentially distributed interarrivals, and with sojourn times which are exponentially distributed at rate $\mu$. Each parking space is considered as a server. When all spaces are occupied, any arriving vehicle must leave and is lost to the system.
\subsection{Setting (Assumptions)}
\begin{itemize}
  \item The system we study consists of a simplified parking lot with pairs of parking spots, shown vertically below, and access lanes on both sides. In the general case, there are $n$ parking spaces with $m$ rows of 2 cars per row so $n=2m$. We will work with a particular case with $n=10$ parking spots in $m=5$ rows. For example, matrix $A$ represents a specific parking situation in the parking lot. The larger black section for each vehicle represents the front of the car. We use the notation $1$ for parking forward, $-1$ for reverse parking, $0$ for empty.
 
 \begin{figure}[htb]
  	\centering
  	\subfigure[matrix]{
  	\begin{minipage}{3cm}
  		\centering
  		\includegraphics[width=3cm,height=3cm]{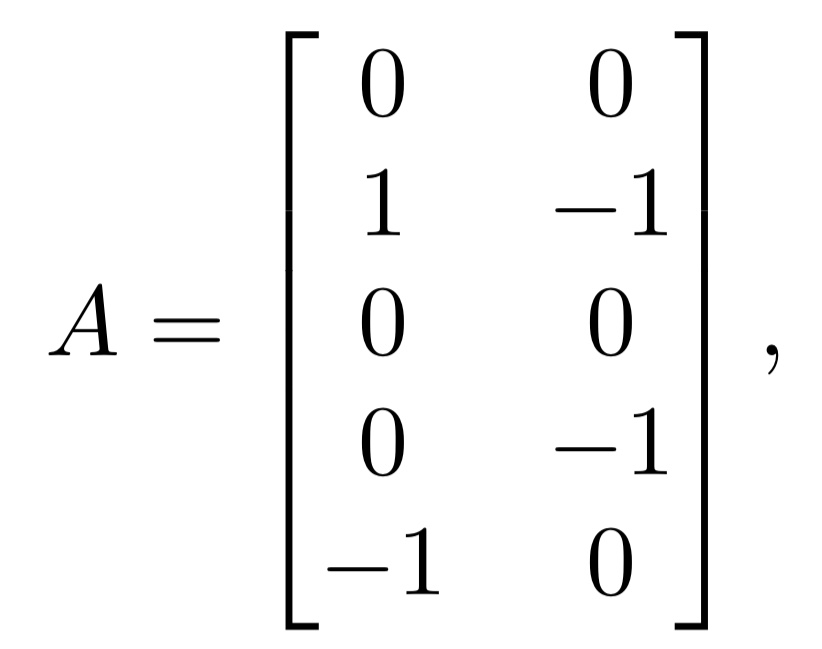}
  	\end{minipage}
  	}
  	\subfigure[parking lot]{
  	\begin{minipage}{2cm}
  		\centering
  		\includegraphics[width=2cm,height=3cm]{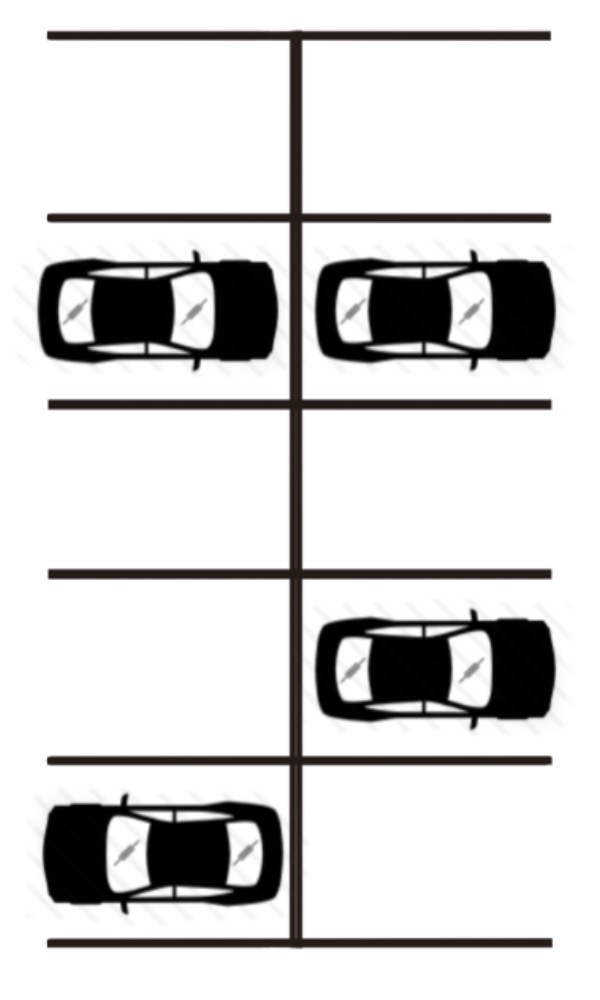}
  	\end{minipage}
  	}
 \end{figure}

  \item Vehicle arrivals occur at rate $\lambda$ according to a Poisson process. Parking times are exponentially distributed with parameter $\mu$.  This parking lot can be viewed as a multi-server queueing system, with vehicles being the customers and the parking spots acting as servers. 
    \item The system runs for time  $T$. In our model, we assume $T=24$ hours.
  \item If there is no free parking space, an arriving vehicle leaves the parking lot and is lost to the system.
  \item There are two ways that a vehicle can drive into a reverse parking or backward parking position. One way is for the vehicle to back into an empty parking position. Another way is the for the vehicle to select an empty pair of spots and to drive through one spot to position itself in a ``reverse'' parking position in the other spot.  
  \item There are different probabilities for different parking spaces: when a parking space is occupied, the parking probability of an arrival vehicle is $0$; when there is a parking spot with no car in the opposite pair position, we choose the parking probability to be twice as large as when there is a car opposite. (This ratio can be treated as a parameter and adjusted as desired.) For our example, the elements in matrix $B$ represent the parking probability of each space in matrix $A$:
  $$B= \begin{bmatrix}
                           0.2 & 0.2 \\
                           0&0 \\
                           0.2 & 0.2  \\
                           0.1 & 0 \\
                           0 & 0.1
                     \end{bmatrix}$$
  \item For a specific parking space, we assume the probability of parking forward is 0.8, backward is 0.2 when the opposite pair position is occupied. We assume that the probability of parking forward is 0.1, backward is 0.9 when the opposite pair position is empty. These probabilities can be taken to be parameters of the model and can be changed according to a particular setting. 
\end{itemize}

\subsection{Markov Analysis}
An M/M/c/c queueing system has a state space $\{0,1,2,3...c\}$ where the state represents the number of vehicles in the system. The following is well known. (\cite{SzJ}, \cite{AdanR})

`` The state-space diagram for M/M/c/c is shown as below:
\begin{figure}[htb]
  \centering
  \includegraphics[scale=0.5]{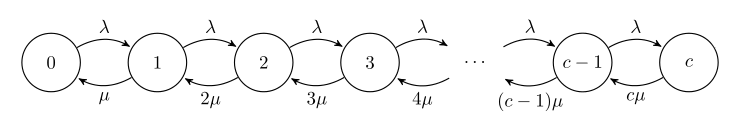}
  \caption{Process diagram for M/M/c/c Queue}
\end{figure}

The transition rate matrix is:
$$Q=\begin{bmatrix}
  -\lambda &\lambda  \\
  \mu &-(\mu+\lambda) &\lambda & \\
  &2\mu &-(2\mu+\lambda) &\lambda \\
  & &3\mu &-(3\mu+\lambda) &\lambda \\
  & & & &\ddots  \\
  & & & &(c-1)\mu &-(\lambda+(c-1)\mu) &\lambda \\
  & & & & &c\mu &-c\mu \\
\end{bmatrix}$$

The steady-state distribution exists since the process has a finite state space. The stationary distribution is 
$$P_k=P_0\left(\frac{\lambda}{\mu}\right)^{k}\frac{1}{k!}=P_0 \rho^{k}\frac{c^{k}}{k!},k=1,2,\cdots,c$$
$$P_0=\left(\sum_{k=0}^c\left(\frac{\lambda}{\mu}\right)^{k}\frac{1}{k!}\right)^{-1}=\left(\sum_{k=0}^c\frac{c^{k}}{k!}\rho^{k}\right)^{-1},$$
where $\rho=\frac{\lambda}{c\mu}$."

In our parking lot model where $c=10$, we illustrate the values of $P_k,k=0,1,2,\cdots,10$ if $\lambda=1$, $\mu=1$.  Results are as shown in Table 1, and our simulated results give close approximations. 

\begin{table}[htb]
\centering
\caption{Calculated and simulated values of $P_i$}
\label{tab:1}       
\begin{tabular}[l]{@{}lcccccc}
\hline
$P_i$ & $P_0$ & $P_1$ & $P_2$ & $P_3$ & $P_4$ & $P_5$\\
\hline
Calculated & $.368$ & $.368$ & $.184$ & $.0613$ & $.0153$ & $.00307$ \\
Simulated & $.368$ & $.367$ & $.184$ & $.0617$ & $.0153$ & $.00311$  \\
\hline
$P_i$ & $P_6$ & $P_7$ & $P_8$ & $P_9$ & $P_{10}$\\
\hline
Calculated & $.000511$ & $.00007$ & $.000009$ & $0$ & $0$\\ 
Simulated & $.000492$ & $.00008$ & $.000008$ & $0$ & $0$\\
\hline
\end{tabular}
\end{table}

At this point, the analysis does not consider the orientation of the vehicle (forward or reverse), but gives a summary of results in the nonoriented model, and helps to emphasize the complexity of the oriented model, and partially confirm the validity of the oriented model.

\subsection{Oriented Model}

We have six different types of oriented paris for each row of the parking lot matrix: Single Forward, Single Backward, One Forward and One Backward, Double Forward, Double Backward, or Empty, as the figure shows below (here written horizontally):

\begin{figure}[htb]
  \centering
  \includegraphics[scale=0.3]{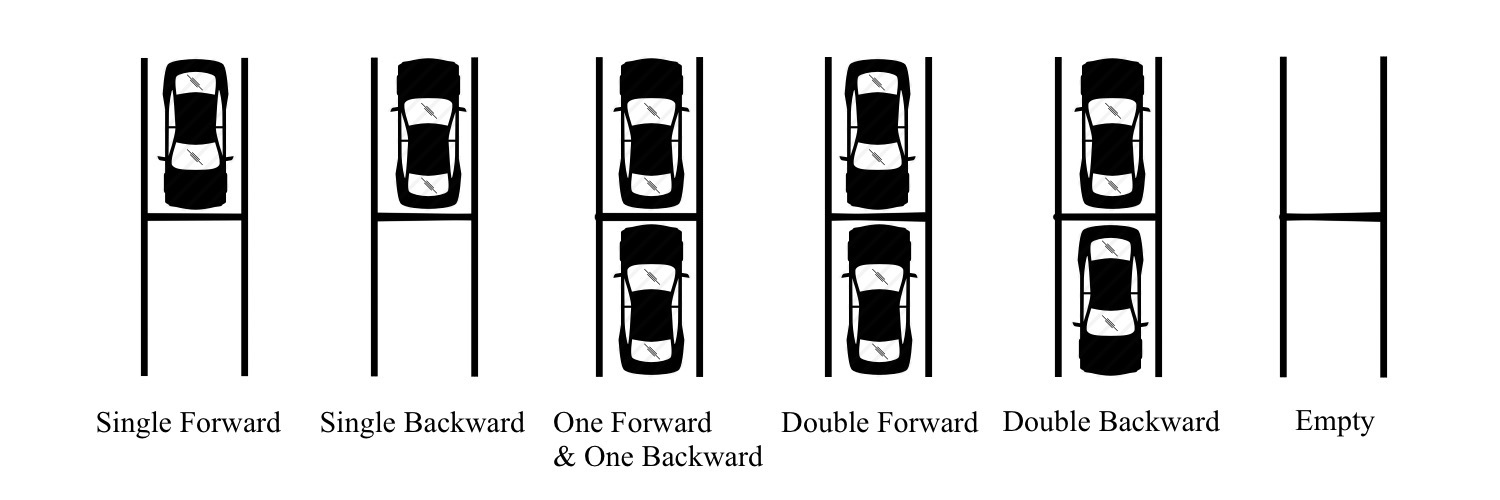}
\end{figure}

Now we want to calculate the number of possible states for our oriented parking lot model with $n$ spots where $n=2m$ with $m$ rows of 2 spots (pairs), where each state is a six tuple with the six components giving the count of the six possible pair types. We are only concerned with the counts of each orientation type and not the order in which they occur. 
\begin{Theorem}
The number of states (6 tuples) in an oriented parking lot with $n=2m$ spots consisting of $m$ pairs is $\binom{m+5}{5}$. 
\end{Theorem}
\begin{proof}
Let $a_1,a_2,a_3,a_4,a_5,a_6$ represents the count of each pair type. The total number of pairs is $m$ so 
  $$\begin{cases}
	a_1+a_2+a_3+a_4+a_5+a_6=m\\
	a_i\ge 0,\quad i=1,2,\cdots ,6
\end{cases}$$
So the number of states is the number of integer solutions. We use the standard combinatorial method. (\cite{FellW}):\\
In our example, with $n=10$, we have $m=5$. 
we represent the pairs by stars. 
\begin{figure}[htb]
  \centering
  \includegraphics[scale=0.3]{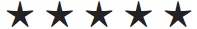}
\end{figure}

We add 1 to each $a_i$ and let $b_i=a_1+1$. Then \\
  $$\begin{cases}
	b_1+b_2+b_3+b_4+b_5+b_6=m+6\\
	b_i\ge 1,\quad i=1,2,\cdots ,6
\end{cases}$$
Adding 1 to each of the six $a_i$ results in $m+6=5+6=11$ stars.  \\
\begin{figure}[htb]
  \centering
  \includegraphics[scale=0.3]{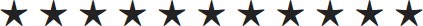}
\end{figure}

To find the number of solutions for the $b_i$ we have to insert $k-1$ separating bars in the spaces between the stars to get $k$ groups of at least one star each. In our model $k=6$, since there are 6 possible orientations.  
\begin{figure}[htb]
  \centering
  \includegraphics[scale=0.3]{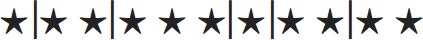}
\end{figure}

Thus the number of solutions for the equation $b_1+b_2+b_3+b_4+b_5+b_6=m+6$ with $b_i\geq 1$ is the number of choices for insertion of the $k-1=6-1=5$ separating bars.  So the number of solutions is $\binom{m+5}{5}$.
This also gives the number of solutions of $a_1+a_2+a_3+a_4+a_5+a_6=m$ with $a_i\geq 0$, and the result follows. 
\end{proof}

We can see how quickly the number of states increases.

\begin{table}[htb]
\centering
\caption{Counting the number of possible states}
\label{tab:2}       
\begin{tabular}[l]{|l|c|c|}
\hline
$n$ & $m$ & $\#\, of\, states$ \\
\hline
$2$ & $1$ & $6$  \\
\hline
4 & 2 & 21\\
\hline
6  &3 & 56\\
\hline
8 &4 &126\\
\hline
10&5 &252\\
\hline
\end{tabular}
\end{table}

In our case, with $n=10$ parking spaces in $m=5$ rows, there are $\tbinom{10}{5}=252$ possible oriented parking configurations.  This is large so it makes sense to move to simulation. The transitions between states are somewhat complex and we illustrate by looking at a particular state $(1,0,0,0,0,4)$, which represents 
1 Single Forward, 0 Single Backward, 0 (One Forward and One Backward), 0 (Double Forward), 0 (Double Backward), and 4 Empty pairs.  A departure of a vehicle occurs with rate $1\mu$ since there is only 1 vehicle in the parking lot and the new state will be $(0,0,0,0,0,5)$. Arrivals to the state $(1,0,0,0,0,4)$ will transition to one of the states 
$(2,0,0,0,0,3)$, $(0,0,0,1,0,4)$, $(0,0,1,0,0,4)$ and $(1,1,0,0,0,3)$. If we consider the transition from $(1,0,0,0,0,4)$ to 
$(2,0,0,0,0,3)$, we note that we start with a single forward vehicle and all other spots empty. The arrival rate is $\lambda$. There are nine empty spots and the probability of choosing a specific empty spot from an empty pair is double that of choosing the empty spot opposite the currently parked vehicle. Thus the probability of each of the 8 empty spots within empty  pairs is 2/17 and the probability of the 1 empty spot opposite the parked vehicle is 1/17. So the rate of moving into a specific empty spot in an empty pair is $(2/17)\lambda$. By our assumptions, given that we choose such as spot, the probability of forward parking is .1 so the transition rate into a specific one of the eight spots is $(2/17).1\lambda$. Since there are eight such spots, the total transition rate from $(1,0,0,0,0,4)$ to $(2,0,0,0,0,3)$ is $(16/17).1\lambda $. 

Similarly, we compute the rates from $(1,0,0,0,0,4)$ to  $(0,0,0,1,0,4)$, $(0,0,1,0,0,4)$ and $(1,1,0,0,0,3)$
as $\frac{1}{17}*0.8\lambda$, $\frac{1}{17}*0.2\lambda$, $\frac{16}{17}*0.9\lambda$, respectively.

\begin{table}[htb]
\centering
\caption{Transition rates from $(1,0,0,0,0,4)$}
\label{tab:3}       
\begin{tabular}[l]{|l|c|c|}
\hline
From & To & Rate\\
\hline
$(1,0,0,0,0,4)$ & $(0,0,0,0,0,5)$ & $\mu$  \\
\hline
$(1,0,0,0,0,4)$ & $(2,0,0,0,0,3)$ & $\frac{16}{17}0.1\lambda$\\
\hline
$(1,0,0,0,0,4)$  &$(0,0,0,1,0,4)$3 & $\frac{1}{17}*0.8\lambda$\\
\hline
$(1,0,0,0,0,4)$ &$(0,0,1,0,0,4)$4 &$\frac{1}{17}*0.2\lambda$\\
\hline
$(1,0,0,0,0,4)$ &$(1,1,0,0,0,3)$5 &$\frac{16}{17}*0.9\lambda$\\
\hline
\end{tabular}
\end{table}

These rates would be the positive values corresponding to the row of the $252\times252$ matrix transitioning from state $(1,0,0,0,0,4)$. Even though it is possible to write a infinitesimal generator matrix for these 252 cases, we choose to use simulation because of its large size and complex calculation.

\section{Pseudo Code}
In order to create an algorithm with fast operating speed, we take advanatge of exponential distribution properties.  We begin with our $5\times 2$ matrix $A$ initially with all zeros. 
\begin{itemize}
  \item We use the \textit{$memoryless \ property$} of exponential distribution and \textbf{Theorem} \ref{T:min} to generate the time until the next event which is exponentially distributed with parameter ${\lambda+\mu*sum(|A|)}$. The more general and more complex method would begin by generating two sequences of exponentially distributed values for interarrival times and service times. 
    \item According to \textbf{Theorem} \ref{T:small}, we decide the type (arrival or service completion) of the next event by generating a random number $r\in (0,1)$ and comparing it to $\frac{\lambda}{\lambda+\mu*sum(|A|)}$, because      the arrival interval and service times follow the exponential distribution with parameters $\lambda$ and $\mu*sum(|A|)$, respectively.
\end{itemize}

Therefore, our pseudo code is shown below:
\begin{algorithm}
  \caption{ Algorithm for getting relation among parking direction and arrival time, service time.}
  \begin{algorithmic}[1]
    \Require
      Arrival time parameter,$\lambda$;
    \Ensure
      Average of the proportion of reversing parking vehicles in the number of parking vehicles, $mean$;
    \Function{Myfunction}{$i$}
    \State $\mu \gets 1$, $Time \gets 0$
    \State $A \gets \begin{bmatrix}
                           0&0 \\
                           0&0 \\
                           0&0 \\
                           0&0 \\
                           0&0
                     \end{bmatrix}$
    \State $mean \gets 0$
    \Repeat
      \State Generate the time point at which the next event occurs, $t$, which has an exponential distribution with parameter $\dfrac{\lambda}{\lambda+\mu*sum(abs(A))}$.
      \State Generate a random number $r$ to determine whether the next event is arrival or departure. If $r<\dfrac{\lambda}{\lambda+\mu*sum(abs(A))}$, the next event is arrival; if $r> \dfrac{\lambda}{\lambda+\mu*sum(abs(A))}$, the next event is departure.
      \If{the next event is arrival}
      \State Create matrix $B$ to indicate the probability of selecting each parking space. 
      \State Generate random number $r$ to select the parking space.
      \State Create matrix $P$ to indicate the probability of parking forward for each space.
      \State Generate random number $r$ to determine whether parking is forward or backward.
      \Else{ Generate random number $r$ to select the parking space from which the vehicle leaves.}
      \EndIf
      \State Update A. Store the proportion of reverse parking in vector $v$.
      \State Time=Time+t.
    \Until{$Time > 24$}
    \State Select the last entry of the the vector $v$ as $prop$.
    \State
    \Return $prop$.
    \EndFunction
  \end{algorithmic}
\end{algorithm}
\clearpage

\section{Simulation and Results}

We assume that $\lambda$ vehicles arrive per hour, the average parking time per vehicle is $\mu=1$ hour.  The important value here is the ratio of $\lambda$ to $\mu$ so we can choose $\mu=1$ without loss of generality by simply using a different choice of time units. 

Let $F$ represent the number of forward facing vehicles and $B$ represent the number of backward facing vehicles at the end of a single simulation run. After using Monte Carlo simulation for the parking lot with the above code, we get the proportion of backward parking as a function of $\lambda$ shown below, where $E(\frac{B}{F+B})$ represents for the average of the proportion of reverse parking in all parking vehicles. This assumes that the system is in steady state, so that the time measured from the start time is large. 

\begin{figure}[htb]
  \centering
  \includegraphics[scale=0.43]{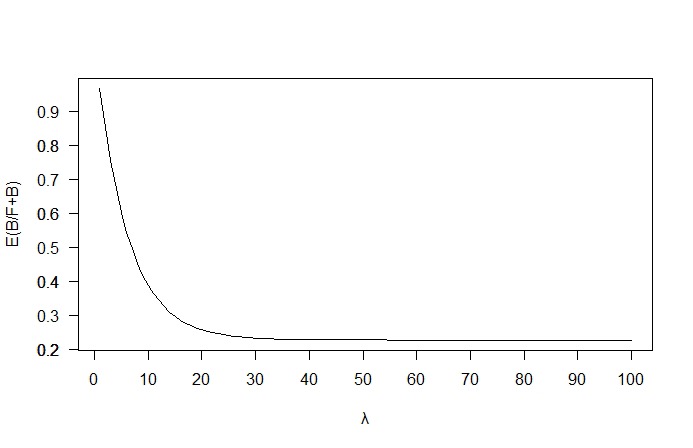}
  \caption{Proportion of Reverse Parking vs Lambda}
\end{figure}

From this image we can see that our oriented model is reasonable. When $\lambda$ is close to $0$, vehicles arrive at a slow rate which makes the parking lot often almost empty, which means the drivers can drive through directly so that the vehicle is automatically in a ``reverse'' status or drivers are more willing to spend some time and effort to reverse park, as there is likely little interference from other vehicles. Thus $E(\frac{B}{F+B})$ approaches $0.9$ (which is our choice of the reverse probability parameter when the both components of a pair are empty). When the $\lambda$ is large, vehicles arrive very frequently to make the parking lot often full, which means the drivers prefer to park forward because the parking lot is busy and crowded. i.e. $E(\frac{B}{F+B})$ approaches 0.2 (which is our choice of the reverse parking probability parameter when one side of a pair is already full).

To illustrate further, we counted the five most frequent states when $\lambda=2$, $\lambda=5$, $\lambda=8$, $\lambda=10$. Recall that a state has six components representing the counts for the six cases 
Single Forward, Single Backward, One Forward and One Backward, Double Forward, Double Backward, or Empty.  

\begin{table}[H]
\centering
\caption{The Five Highest Frequency states}
\label{tab:2}       
\begin{tabular}[l]{@{}lccccc}
\hline
$\lambda$ & $1st$ & $2nd$ & $3rd$ & $4th$ & $5th$ \\
\hline
$\lambda=2$ & $(0,2,0,0,0,3)$ & $(0,4,0,0,0,1)$ & $(0,4,1,0,0,0)$ & $(0,3,0,0,0,2)$ & $(0,1,1,0,0,3)$\\
\hline
$\lambda=5$ & $(0,2,0,0,0,3)$ & $(0,4,1,0,0,0)$ & $(0,2,0,1,0,2)$ & $(1,2,0,0,0,2)$ & $(0,3,0,0,0,2)$ \\ 
\hline
$\lambda=8$ & $(0,2,1,0,0,2)$ & $(0,3,0,0,1,1)$ & $(1,0,0,0,0,4)$ & $(0,3,1,0,1,0)$ & $(0,0,1,0,0,4)$ \\
\hline
$\lambda=10$ & $(0,0,4,1,0,0)$ & $(0,0,3,2,0,0)$ & $(0,1,3,1,0,0)$ & $(0,1,4,0,0,0)$ & $(0,1,2,2,0,0)$\\
\hline
\end{tabular}
\end{table}
From the table we can see that as $\lambda$ increases, that is, the number of vehicles entering the parking lot per hour increases, more Double Forward counts appear in the parking lot.  

We indicated that Figure 2 represents the proportion of reverse parking for the system in steady state.  
To numerically justify that claim, we consider the proportion as a function of time for several values of $\lambda$. 
\begin{figure}[H]
  \centering
  \includegraphics[scale=0.43]{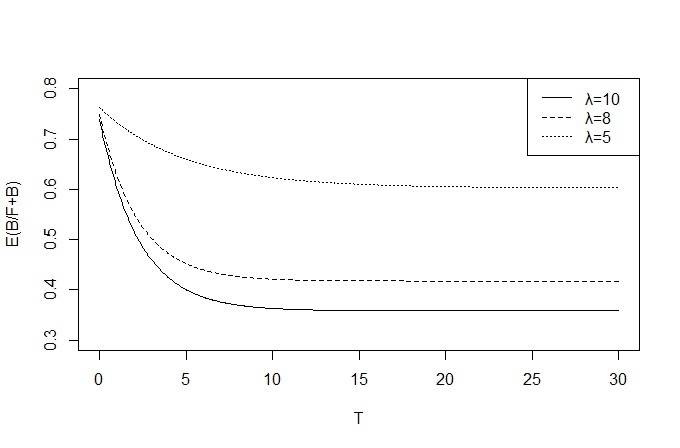}
  \caption{Result}
\end{figure}
Figure $3$ depicts $E(\frac{B}{F+B})$ as a function of time when $\lambda=5,8,10$. These three curves tend to level off to values $0.62,0.42,0.38$, beyond  $T=10$. These values are consistent with the values that appear in Figure 2 for the 3 values of $\lambda$ used. Further, we notice that when the lambda is larger, the time it takes for the parking lot reverse parking proportion to stabilize will be shorter.

\section{Conclusion and Future Research}
According to our model, we find that the proportion of reverse parking for parked vehicles is related to the arrival rate and  sojourn time of the vehicle. The larger the arrival rate, the lower the proportion of reverse parking will be. This model can help to explain the observed proportions of reverse parked vehicles in parking lots. If it is considered desirable to have a higher proportion of reverse parked vehicles, our model may suggest ways to achieve this. For example, more parking spots may mean a higher proportion of empty pairs in the lot which would encourage drivers to use reverse parking. Overall, our model can help the parking lot to develop a parking policy based on different vehicle entry and exit rates.

Further research on this topic is very possible. Our choice of probability values implies that the two possible access lanes are equally probable. In many cases, one access lane may be more convenient so probabilities could be adjusted to reflect that condition. Further, arrival rates will often change during the course of the day, so time dependent arrival rates could be added to the model to give more realistic results in many settings. 

{\bf Acknowledgements.} We acknowledge funding and support  from MITACS Global Internship program, CSC scholarship. .

\end{document}